\documentclass[envcountsame,runningheads]{llncs}

\usepackage[T1]{fontenc}
\usepackage{graphicx}
\usepackage{amsfonts}       
\usepackage{amsmath}		
\usepackage{cleveref}       
\usepackage{mathtools}
\usepackage[ruled,lined,linesnumbered, noend]{algorithm2e}
\usepackage[dvipsnames]{xcolor}
\usepackage{longtable}

\usepackage{pscproc2}

\newcommand{\bigO}{\mathcal{O}}

\newcommand{\dollar}{\$}
\newcommand{\hashtag}{\#}
\newcommand{\minlength}{h}

\newcommand{\href}[2]{#2}

\begin{document}

\title{String Partition for Building Long Burrows-Wheeler Transforms}
\author{Enno Adler \and
Stefan Böttcher \and
Rita Hartel}
\authorrunning{Adler et al.}
\titlerunning{String Partition for Building Long BWTs}
\institute{Paderborn University,\\ 
\email{\{enno.adler, stefan.boettcher, rita.hartel\}@uni-paderborn.de}
}
\maketitle

\begin{abstract}
    Constructing the Burrows-Wheeler transform (BWT) for long strings poses significant challenges regarding construction time and memory usage. We use a prefix of the suffix array to partition a long string into shorter substrings, thereby enabling the use of multi-string BWT construction algorithms to process these partitions fast. We provide an implementation, partDNA, for DNA sequences. Through comparison with state-of-the-art BWT construction algorithms, we show that partDNA with IBB offers a novel trade-off for construction time and memory usage for BWT construction on real genome datasets. Beyond this, the proposed partitioning strategy is applicable to strings of any alphabet.
\end{abstract}


\section{Introduction}

The Burrows-Wheeler transform (BWT)~\cite{burrows1994} is a widely used reversible string transformation with applications in text compression~\cite{ferragina2005}, indexing~\cite{ferragina2005}, and short-read alignment~\cite{langmead2012}. The BWT reduces the number of equal-symbol runs for data compressed with run-length encoding and allows pattern search in time proportional to the pattern length~\cite{ferragina2005}. Because of these advantages, and the property that the BWT of a string $S$ can be constructed and reverted in $\bigO(|S|)$ time and space, the BWT is important in computational biology.

In computational biology, there are at least two major BWT construction cases: The BWT can be generated either for a single long genome, for example, the reference genome~\cite{langmead2012}, or for a collection of comparatively short reads~\cite{bauer2013}. Like BCR~\cite{bauer2013} and ropebwt2~\cite{Li2014}, the standard way~\cite{cenzato2022} to define and compute a text index for a collection $W$ of strings $W_i$ is to concatenate the $W_i$ with different end-marker symbols $\hashtag_i$ between the $W_i$: $W'=W_0\hashtag_0W_1\hashtag_1\dots{}W_{k}\hashtag_k$.\footnote{The $\hashtag_i$ symbol at the end of the strings $W_i$ and the $\dollar$ symbol in $S$ are only different to explain the concept and break ties; implementations like ropebwt2~\cite{Li2014} and our approach, partDNA, use the same symbol for each end-marker.} The multi-string BWT~\cite{egidi2019} we use is also called BCR BWT~\cite{navarro2023}, or mdolBWT~\cite{cenzato2022}. 

In this paper, we show how to partition a string $S$ into a collection $W$ of short strings $W_{i}$ and order the $W_i$ in such a way that the BWT of $W$ is similar to the BWT of $S$. We say `similar' here because the BWT of $W_0, \dots, W_k$ contains $k+1$ $\hashtag$ symbols that we need to remove from the BWT after construction. For example, compare the BWT of $S = CAAAACAAACCGTAAAACAAACCGGAACAA\dollar$ to the BWT of the collection $W = \{W_0, \dots, W_8\}$ of words 

\begin{align*}
    &W_0 = A, 
    W_1 = A, 
    W_2 = AAACCGGAAC,
    W_3 = AAACCGT, \\
    &W_4 = \$C,
    W_5 = A,
    W_6 = A, 
    W_7 = AAAC, 
    W_8 = AAAC. 
\end{align*}
Because $W_4 W_6 W_8 W_3 W_5 W_7 W_2 W_1 W_0 \cdot \$ = \$ \cdot S$, $W$ is a partition of the first right rotation\footnote{Using the substring notation of Section~\ref{section:preliminaries}, the first right rotation of $S$ is $\dollar{} \cdot S[0,|S|-1]$.} of $S$. The BWT of $W$ is constructed using $W_0, \dots, W_k$ in the order of their indices:
\begin{align*}
BWT(W) &= AACTCAACC\hashtag\hashtag\hashtag\hashtag\hashtag\hashtag\hashtag\hashtag\hashtag{}GAAAAAAAAAA\$AAAACCGCCG \\
BWT(S) &=
AACTCAACC\phantom{\hashtag\hashtag\hashtag\hashtag\hashtag\hashtag\hashtag\hashtag\hashtag{}}GAAAAAAAAAA\$AAAACCGCCG
\end{align*}
If we remove the run of $\hashtag$ symbols, the BWTs of $S$ and $W$ are identical. In Figure~\ref{figure:concept}, we show the connection between single-string and multi-string BWT construction algorithms and the contribution of our paper. 

\begin{figure}
    \centering
    \includegraphics{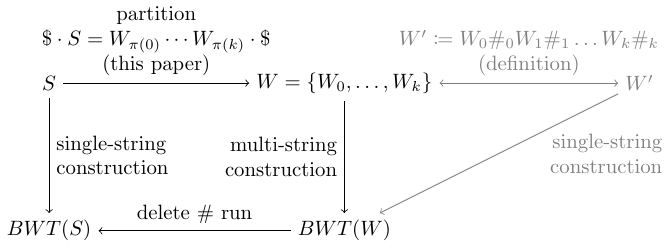}
    \caption{Summary of the relationship of single-string and multi-string BWT construction algorithms and the contribution of this paper. The $BWT$ for $S$ can also be obtained by partitioning $S$, using a multi-string construction algorithm on the sorted partition, and removing the $\hashtag$-run at the end. The output of the multi-string BWT construction algorithm is equal to the BWT for $W'$. }\label{figure:concept}
\end{figure}

Herein, our main contributions are as follows:

\begin{itemize}
    \item The proof of the correctness of computing the $BWT(S)$ as shown in Figure~\ref{figure:concept}. 
    \item An implementation, partDNA, to partition long DNA sequences for BWT construction.\footnote{The implementation is available at \texttt{https://github.com/adlerenno/partDNA}.} 
    \item A comparison of state-of-the-art BWT construction algorithms regarding the construction time and memory usage of using the partition.
\end{itemize}
\section[Related Work]{Related Work}

The BWT~\cite{burrows1994} is fundamental to many applications in bioinformatics such as short-read alignment. Bauer et al.~\cite{bauer2013} designed BCR for a collection of short DNA reads. BCR inserts all reads in parallel starting at the end of each sequence at the same time. The position to insert the next symbol of each sequence is calculated from the current position and the constructed part of the BWT.\@ The BWT is partitioned into buckets where one bucket contains all BWT symbols of suffixes starting with the same character. The buckets are saved on disk. Ropebwt, ropebwt2 by Li~\cite{Li2014}, and IBB by Adler~\cite{adler2025} are similar to BCR, but use $B+$ trees instead of linear saved buckets. 

The BWT can be obtained from the suffix array by taking the characters at the position before the suffix. Thereby, suffix array construction algorithms (SACAs) like divsufsort~\cite{fischer2017}, SA-IS by Nong et al.~\cite{nong2011}, libsais, gSACA-K by Louza et al.~\cite{Louza2016}, or gsufsort by Louza et al.~\cite{Louza2020} can be used to obtain the BWT.\@ Many SACAs rely on induced suffix sorting. Induced suffix sorting derives the order of the previous positions of a sorted set of suffix positions if these point to equal characters. An example is shown in Figure~\ref{figure:buckets}.

The grlBWT method by D{\'{\i}}az{-}Dom{\'{\i}}nguez et al.~\cite{navarro2023} uses induced suffix sorting, but additionally uses run-length encoding and grammar compression to store intermediate results and to speed up computations required for BWT construction.

eGap by Egidi et al.~\cite{egidi2019} divides the input collection into small subcollections and uses gSACA-K~\cite{Louza2016} to compute the BWT for the subcollections. Thereafter, eGap merges the subcollection BWTs into one BWT.\@

In a prefix-free set, no two words from the set are prefixes of each other; thus, their order is based on a different character rather than on word length. Consequently, the order of two suffixes starting with words from the prefix-free set is determined by these words. BigBWT by Boucher et al.~\cite{boucher2019} builds the BWT using prefix-free sets. BigBWT replaces the input with a dictionary D and a parse P using the rolling Karp-Rabin hash. P is the list of entries in D according to the input string; D forms a prefix-free set.
r-pfbwt by Oliva et al.~\cite{oliva2023} recursively uses prefix-free parsing on the parse P to further reduce the space needed to represent the input string.

SA-IS~\cite{nong2011}, libsais, gSACA-K~\cite{Louza2016}, BigBWT~\cite{boucher2019}, and r-pfbwt~\cite{oliva2023} partition the input at LMS-Positions or by using a dictionary. The difference of our partition to all these approaches is that we do not create a BWT as a result. Instead, we translate the problem into a multi-string BWT construction problem and compute the BWT using such a construction method.

\section{Preliminaries} \label{section:preliminaries}

In Appendix~\ref{appendix:abbreviations}, there is a list of symbols used in this paper.

We define a string $S$ of length $|S| = n$ over $\Sigma$ by $S = a_0 a_1 \cdots a_{n-1}$ with $a_i \in \Sigma$ for $i < n$ and always append $a_{n} = \$$ to $S$. We write $S[i] = a_i$, $S[i, j] = a_i a_{i+1} \cdots a_j$ for a substring of $S$, and $S[i..] = S[i, n]$ for the suffix starting at position $i$. For simplicity, we assume that $S[-1] = S[n] = \$,$ and also allow $S[-1, j] = a_{n} a_0 \cdots a_j$ as a valid interval.

The suffix array $SA(S)$~\cite{manber90} of a string $S$ is a permutation of $\{0, \ldots, n\}$ such that the $i$-th smallest suffix of $S$ is $S[SA(S)[i]..]$. The suffix array $SA(W)$ and document array $DA(W)$ for a collection $W$ of strings $W_i$ are arrays of numbers such that the $j$-th smallest suffix of all $W_i$ is $W_{DA(W)[j]}[SA(W)[j]..]$.



The Burrows–Wheeler transform $BWT(S)$ of a string $S$~\cite{burrows1994} can either be obtained by \(BWT(S)[i]=S[SA(S)[i]-1]\) or by taking the last column of the sorted rotations of $S$.

\section{Partition Theorem}

A single string $S$ has only unique suffixes because the suffixes differ in their length. However, if we partition $S$ into $W = \{W_0, \dots, W_k\}$, $W_0, \dots, W_k$ can have several equal suffixes. For example, $W_2 = AAACCGGAAC$ and $W_7 = AAAC$ both have the suffix $AAC$. Given only the suffixes, we cannot decide how to order the characters $G$ and $A$ in $BWT(W)$ that occur before the suffixes $AAC$ in $W_2$ and $W_7$. We break the tie by using the word order of $W_2$ and $W_7$. For that purpose, we choose the word order of the $W_i$ to be the order of the suffixes occurring in $S$ after these words $W_i$. In other words, the index $i$ of word $W_i$ is the number of strictly smaller suffixes within the set of all suffixes of $S$ that start behind words from the collection $W$. In Figure~\ref{figure:psa}, we visualized this concept of obtaining the word indices.

We transfer this argumentation from suffixes to suffix arrays: The index $i$ of word $W_i$ is the index $i$ within a filtered suffix array that contains only those positions in $S$ that are after the words of $W$. Computing the full suffix array and thereafter filtering it for the positions that follow the word ends yields the correct order of the words. However, if we would construct the full suffix array to partition $S$, this would be inefficient because we can obtain the BWT from the suffix array directly. But this trivial solution shows that we can obtain the word indices in $\bigO(n)$.

In multi-string construction of $BWT(W)$, the last characters $c_0, \dots, c_k$ of each word $W_0, \dots, W_k$ are inserted at the first positions of the BWT of $W$; thus, for our approach, $c_0, \dots, c_k$ must be the first $k+1$ symbols of $BWT(S)$. Because $BWT(S)[i]=S[SA(S)[i]-1]$, the positions in $S$ following the words $W_0, \dots, W_k$ form a continuous subarray at the lowest positions of the suffix array of $S$: A prefix of the suffix array. 

We define $PSA(S) = SA(S)[0\dots k]$ as a prefix of length $k+1$ of the suffix array of $S$. In the following, we assume $k < n$ to be a fixed value. We write $t \in PSA(S)$, if there exists $i$, $0 \leq i \leq k$, such that $t = PSA(S)[i]$.

\begin{figure*}
    \centering
    \includegraphics[width=\linewidth]{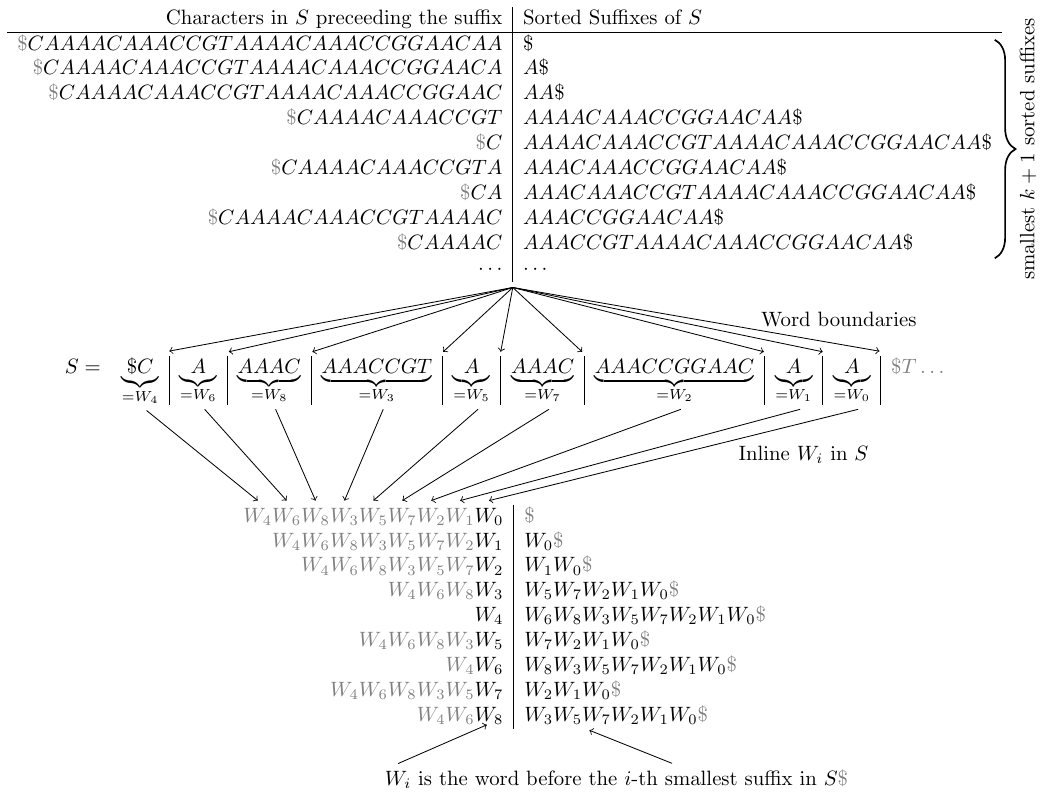} 
    \caption{Concept of partitioning $S$ using $PSA(S)$ with $k = 8$ to obtain the collection $W$. The suffix of a word $W_i$ in $S$ could either be expressed by characters of $S$ or by words from $W$ in the order of their appearance in $S$}\label{figure:psa}
\end{figure*}

For each suffix entry $PSA(S)[j]$, let $\Omega(j) = \max(\{-1\} \cup \{ t \in PSA(S) : t \leq PSA(S)[j]-1\})$ be the next smaller suffix array entry in $PSA(S)$ or $-1$, if $PSA(S)[j]$ is already the smallest value in $PSA(S)$.
We define $W$ as the collection of words $W_i = S[\Omega(i), PSA(S)[i]-1]$ for $0 \leq i \leq k$, which is a partition of the first right rotation of $S$ into substrings. 

As $W$ is a partition of the first right rotation of $S$, each character of $S$ is mapped to exactly one character in one of the words of $W$. Therefore, we define two functions $word$ and $position$ to map a position $q$ from $S$ to the word $W_{word(q)}$ in $W$ and position $position(q)$ of that character $S[q]$ in $W_{word(q)}$: For each $q$ with $-1 \leq q < |S|$ exists exactly one $j \leq k$ such that $q \in [\Omega(j), PSA(S)[j]-1]$, because each position in the partition belongs to exactly one interval. We set $word(q) = j$ and $position(q) = q - \Omega(j)$. We also set $word(n) = word(-1)$ and $position(n) = position(-1)$.

In the initial example, the $C$ at position $0$ in $S$ corresponds to the $C$ in word $W_4$ at position $1$, so $word(0) = 4$ and $position(0) = 1$.

In the edge cases for $k$ in the $PSA(S)$ definition, we obtain for $k = 0$ that the prefix of the suffix array only contains the position of the $\dollar$ symbol, which is at the end of the word. It follows $\Omega(n) = -1$, because $\{ i \in PSA(S) : i < PSA(S)[i]\} = \emptyset$. Thereby, $W$ contains only one word which is $W_0 = S[\Omega(i), PSA(S)[i]-1] = S[-1, n-1]$, which is the first right rotation of $S$. Using the $\$$ as an end marker, the BWTs of $S$ and $W$ are equal, because the suffixes have not changed.

For the edge case $k = n$, $PSA(S) = SA(S)$ and $\{ i \in PSA(S) : i < PSA(S)[i]-1\} = \{0, \dots, PSA(S)[i] - 1\}$ because the suffix array contains all positions of $S$. Thereby, we get $\Omega(i) = PSA(S)[i] - 1$ for all positions $i$, so 
\begin{align*}
   W_i &= S[\Omega(i), PSA(S)[i]-1] \\
   &= S[PSA(S)[i]-1, PSA(S)[i]-1] \\
   &= S[PSA(S)[i]-1] \\
   &= S[SA(S)[i]-1] \\
   &= BWT(S)[i].
\end{align*}
Each word is one symbol and their order is already the BWT of $S$ because the order is obtained from the suffix array of $S$. If we construct the BWT from the words in this edge case, it is easy to see that $BWT(W) = BWT(S) + \hashtag^{n}$.

\begin{theorem}\label{theorem:partition}
    Let $l (= k + 1)$ be the size of $W$ and $m (= n + l)$ be the total length of $BWT(W)$. 
    Then, for all $i < m$:

    \[ BWT(W)[i] = \begin{cases}
        BWT(S)[i] & 0 \leq i < l \\
        \hashtag & l \leq i < 2l \\
        BWT(S)[i - l] & 2l \leq i < m
    \end{cases} \] 

    Thus, if we know $BWT(W)$, we get by removing the $\hashtag$-run \[BWT(S) = BWT(W)[0,l-1] + BWT(W)[2l,m-1].\]
\end{theorem}
The full proof is in Appendix~\ref{appendix:proof}. Proof sketch: Using the functions $position$ and $word$, we prove the statement: if two suffixes $S[i..] < S[j..]$, then the order of the suffixes is $W_{word(i)}[position(i)..] < W_{word(j)}[position(j)..]$. This can be proven by using $c$, the smallest value for which either $S[i+c] \neq S[j+c]$, or $i+c \in PSA(S)$. Using this statement, we express $SA(W)$ and $DA(W)$ with $SA(S)$ and the $position$ and $word$ functions. Then, we use \(BWT(W)[i]=W_{DA(W)[i]}[SA(W)[i]-1]\) for the most cases to retrieve the $BWT(W)$ from $SA(W)$ and $DA(W)$.

\section{Partition DNA Sequences: partDNA}\label{section:dna_prefix_suffix_array}

Next, we partition a DNA sequence because genomes are long strings over $ACGT$. To partition a DNA sequence $S$, we first find the words having the smallest suffixes in $S$, and second, we order the words according to their suffixes in $S$.



We do not allow every $k$ as the length of the prefix $PSA(S)$ of the suffix array and $k$ is not given explicitly. Instead, the exact value of $k$ will be part of the result of the partition. In particular, we use the chosen length $\minlength$ for the minimal length of an $A$ run as a parameter to partition $S$.

We partition $S = CAAAACAAACCGTAAAACAAACCGGAACAA\$$ with $\minlength=3$ within our following continuous example. The example is visualized in Figure~\ref{figure:example}. Before we go into the details, we give a short high-level description: In Step 1, we scan $S$ to find an initial collection of words. This collection is smaller than the final collection $W$, because we can avoid sorting the larger collection $W$ by using induced sorting. This will be Step 6. To sort the subset of suffixes behind the words in the collection, we sort the initial collection of words in Steps 2 and 3 and if at least two words are equal (Step 4), we will break their tie by their suffixes. As these suffixes start with words again, see Figure~\ref{figure:psa} for an example, and we have already sorted the words in Step 2 and 3, we can use this to assign a unique name to every different word (Step 4) and compute a suffix array (Step 5) that breaks all remaining ties. With the induced sorting in Step 6 and the final polish in Step 7 we get $W$.

\begin{figure*}
    \centering
    \includegraphics[width=\linewidth]{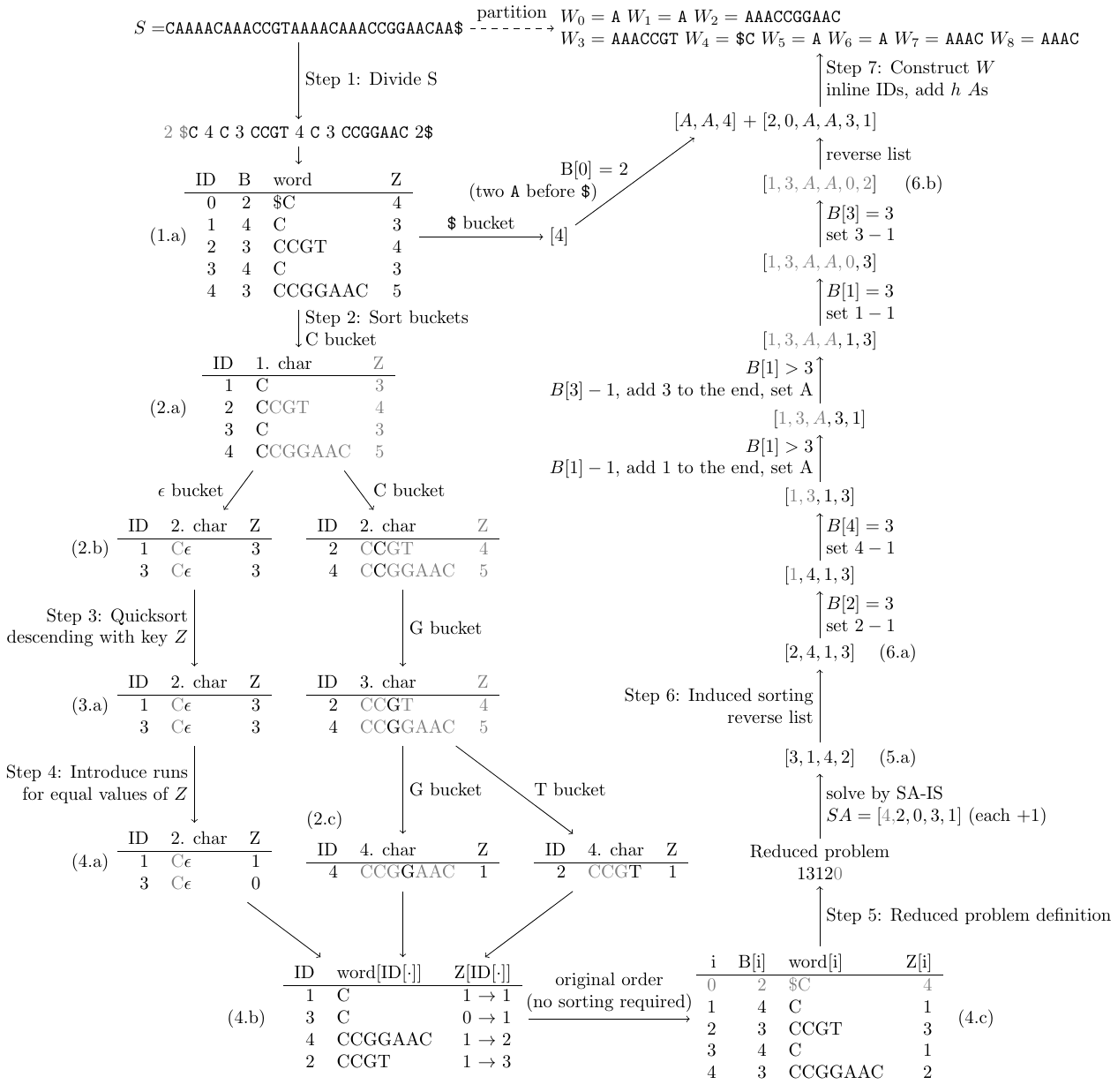} 
    \caption{Example calculation of partitioning using $\minlength = 3$ on the word $S$. In Steps 2 and 3, we only reorder the array ID, the reordering of columns of the other elements is only shown for illustration. Steps 2 and 3 are done in place, so there is no action needed to go from (4.a) and (2.c) to (4.b).}\label{figure:example}
\end{figure*}

In Step 1, we divide $S$ either before each run of at least $\minlength$ $A$ symbols or before $A^*\dollar$ and assign IDs to the words in order of their appearance. Note that these word IDs are different from the word incides that are assigned to order their appearance in $W$. Additionally, we avoid writing runs having $\minlength$ or more $A$s, we only save their run-length. This results in fewer words than the final partition in $W$ has because we obtain the other words by induced suffix sorting in Step 6. We keep the number of $A$s before a word in an array $B$ and the number of $A$s after the word in an array $Z$. When we reach the $\dollar$ symbol, we keep the length of the $A$-run preceeding $\dollar$ in $B[0]$, so in the example of Figure~\ref{figure:example}, $B[0] = 2$. We store the maximum length of an $A$ run plus 1 in the last element of $Z$ regardless of how long the current $A$ run is; so in the example, we store $Z[4] = 5$. We use these values of $Z$ later to sort the words descending because we designed the values in $Z$ in such a way that a higher value implies that a smaller suffix follows. Because $S$ does not start with at least $\minlength$ $A$s, $S$ is not divided at position $0$, which is also the position after the $\dollar$ regarding rotations of $S$. Therefore, we place the $\dollar$ at the beginning of the word with the ID $0$.

In Step 2, we sort the IDs with a value of $1$ and greater (Figure~\ref{figure:example} (1.a)), which form our initial bucket. We sort the IDs by recursively refining buckets: At recursion depth $i$, we group the IDs of one bucket into subbuckets by the $i$-th character of their words. Therefore, the call structure is a tree (Figure~\ref{figure:example} (2.a) to (2.c)), that has a maximum number of 5 branches on the next level, one for each character $A$, $C$, $G$, and $T$ and one branch named $\epsilon$. In particular, to order the IDs into the subbuckets, we first count the number of symbols to get the size of the subbuckets, then create $5$ auxiliary arrays, one for each bucket, copy in a single scan the IDs into the auxiliary array and finally copy all auxiliary arrays back in order. The arrays $B$ and $Z$ and the words are only accessed indirectly and changed by using $B[ID[i]]$ and $Z[ID[i]]$, so we only sort the IDs in Step 2.

If a word $w$ has no further characters, we reached an $\epsilon$-leaf of the sort tree for the word $w$. Then, the suffix of $S$ starting with $w$ is smaller than the suffixes of $S$ of the words that have a character at that position. The reason is that the pattern $A^{y + \minlength}$ or the lexicographically even smaller suffix $A^y\$$ with $y \geq 0$ occurs after $w$ in $S$, because we used these patters to divide $S$ into the words. Because we used different patterns to divide $S$, we sort the words according to the order of these patterns with quicksort as Step 3. In Figure~\ref{figure:example} (1.a), the patterns $A^3$,  $A^4$, and $A^2\$$ are encoded by the values $3$, $4$, and $5$ in the array $Z$. 
Therefore, the quicksort sorts the IDs in descending order by using the values in $Z$ as keys. 

After the quicksort in the leaves of the sort tree, we might have identical words as in Figure~\ref{figure:example} (3.a). In the leaf, two words are equal if they appear consecutively and their value in $Z$ is equal. In Step 4, we encode equal words as a run, so the first of the equal words gets a $1$ and each other word gets a $0$ in $Z$. This allows an easier assignment of names in the following step. All words in leaves containing only one word get a $1$ in $Z$, like in Figure~\ref{figure:example} (2.c).

To solve the case that we found at least two equal words in Step 4, we define a reduced suffix array construction problem to sort them. In Figure~\ref{figure:example} (4.a), the words with IDs 1 and 3 are equal. The reduced problem is defined as in SA-IS~\cite{nong2011} and similar SACAs. In a single pass over $Z[ID[i]]$, we give integer names to the words, like in Figure~\ref{figure:example} (4.b): We add the current value of $Z$, which is $0$ or $1$, to the last name and then assign the sum to $Z$. The smallest assigned name is $1$, because we use $0$ as a global end marker for a valid suffix array problem definition. We obtain the reduced problem by reading the names from $Z[i]$ omitting the first word with ID $0$.

In Figure~\ref{figure:example} after Step 5, the reduced problem is $R = 13120$, with the global end marker $0$. We obtain $SA(R) = [\textcolor{gray}{4,} 2, 0, 3, 1]$. We omit the first value $4$ because the $4$ points towards the end marker $0$. We need to add $1$ to the $SA(R)$ values because we skipped the word with ID $0$ in the recursive problem. We obtain $[3, 1, 4, 2]$ (Figure~\ref{figure:example} (6.a)).

In Step 6, we obtain all words from the sorted IDs by induced suffix sorting.\@ Figure~\ref{figure:buckets} shows the possible induced suffix sorting steps. Let $L$ be the list of IDs in inverted order, so $L = [2, 4, 1, 3]$. We iterate from the start until we reach the end of $L$. For the ID $j$, if $B[j]$ is greater than $\minlength$: We reduce $B[j]$ by 1, write an $A$ to the current position of the list, and append the current ID $j$ to the end of the list again. If $B[j]$ is $\minlength$: We reduce $j$ by $1$ at the current position because we sorted suffixes. We invert the list again, so $L = [2, 0, A, A, 3, 1]$. 

\begin{figure*}
    \centering
    \includegraphics*[scale=1]{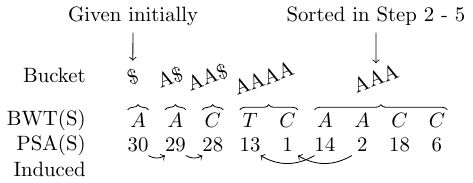}
    \caption{The buckets and the steps of induced suffix sorting that start and end inside the displayed interval. The name of a bucket is a shared prefix of the suffixes. For the $A^v$ buckets it is required that the next symbol is not an $A$, otherwise, the $AAA$ and $AAAA$ bucket would also belong to the $AA$ bucket. We induce the positions from the $\$$ bucket to the next bucket on the right and reduce the suffix array entry by 1 if and only if the symbol at the position before the entry is an $A$, which is the symbol in the BWT row. In the same way, we can fill the $A^v$ buckets right to left.}\label{figure:buckets}
\end{figure*}

For the $\$$-bucket, we also use induced sorting: The starting list is $Q = [4]$ because the word with ID $4$ in Figure~\ref{figure:example} (1.a) of the divided words has $A^{B[0]}\$$ as a suffix and we do $B[0]$ steps each inserting one $A$. The result is $Q = [A, A, 4]$ and the combined result is $Q + L$. 

In Step 7, each $A$ in $Q + L$ yields an own word $W_i = A$. Each ID in $Q + L$ yields the corresponding word in Figure~\ref{figure:example} (1.a). Additionally, in front of each word with an ID that is $1$ or greater in Figure~\ref{figure:example} (1.a), we prepend $\minlength$ $A$s. After Step 7, we get the partition:

\begin{align*}
    &W_0 = A, 
    W_1 = A, 
    W_2 = AAACCGGAAC,
    W_3 = AAACCGT, \\
    &W_4 = \$C,
    W_5 = A,
    W_6 = A, 
    W_7 = AAAC, 
    W_8 = AAAC. 
\end{align*}

To summarize: By partitioning $S$ into the collection $W = \{W_0, \dots, W_8\}$, we have transformed the task of constructing the BWT of a long string $S$ into the task of constructing the BWT of a collection $W$ of smaller words. 

\section{Runtime Complexity}\label{appendix:connection_to_sais}

We claim that partDNA runs in $\bigO(n)$. First, the number of words of the scan in Step 1 is upper bounded by $\frac{n}{\minlength}+1$, because there is always at least one word, which is produced by the $\$$-symbol, and there needs to be an $A$-run of length $\minlength$ or longer between two words to divide the words. Thus, in all next steps, we will work on at most $\frac{n}{\minlength}$ words, because we process the word with ID $0$ differently.

We simplify the recursion to ease the analysis of the Steps 2, 3, and 4. First, instead of using quicksort on the $Z$-values, we extend the word with ID $i$ by a $Z[i]$-long $A$-run. Note that after $\minlength$ consecutive $A$-buckets, the sort order changes in the sense that the $A$ bucket is smaller than the $\epsilon$-bucket and that there are only those two buckets anymore, however this does not change the number of steps to perform the sort. Additionally, we stop the recursive bucket sort only if an $\epsilon$-bucket is reached, we do not stop if it is the last element in a bucket. These adjustments let us remove the quicksort at the higher cost of a deeper recursion. Using this simplification, each letter of $S$ is exactly used once to put an ID into a subbucket, because each letter of $S$ belongs to exactly one word and the $i$-th letter is used only at recursion depth $i$ in Step 2. Putting an ID into a bucket needs only $\bigO(1)$. Thereby, Step 2 is in $\bigO(n)$. Step 4 uses $\bigO(1)$ steps per word in the $\epsilon$-leaf, thus Step 4 needs $\bigO\left(\frac{n}{\minlength}\right)$.

For Step 5, the reduced problem $R$ consists of one letter per sorted word plus the end marker, thus $R$ has size $\frac{n}{\minlength} + 1$. The suffix array of $R$ needs $\bigO(\frac{n}{\minlength})$ using SA-IS. All induced sorting steps together, including those of Step 6, perform $\frac{n}{\minlength}+1$ times an insertion of a word ID and up to $n$ times an insertion of an $A$, because $|S| = n$ limits the number of $A$s in $S$. Hence, $\bigO(n) + \bigO\left(\frac{n}{\minlength} + 1\right) = \bigO(n)$.

As all steps are in $\bigO(n)$, partDNA is in $\bigO(n)$.


\section{Experimental Results}

We compare the BWT construction algorithms on the datasets listed in Table~\ref{tab:bwt_construction_algortihms} regarding construction time and RAM usage. We obtain time and maximum resident set size (max-rss) by the \texttt{/usr/time} command.\footnote{The test is available at \texttt{https://github.com/adlerenno/partDNAtest}.} Each input file is a concatenation of all bases within the reference file because we want to test a single long input string. Ambiguous bases are omitted. In Table~\ref{tab:bwt_construction_algortihms}, if a partition length $\minlength$ is provided, the dataset was partitioned using $\minlength$ from the dataset with the row $\minlength = \text{–}$. We performed all tests on a Debian 5.10.209-2 machine with 128GB RAM and 32 Cores Intel(R) Xeon(R) Platinum 8462Y+ @ 2.80GHz. 

ropebwt3, which uses a SA-IS implementation, BigBWT, r-pfbwt, divsufsort, libsais, grlBWT, and gsufsort compute the BWT of a collection by concatenating the strings. Hence, they are only tested on the single-string dataset (called original), because partitioning the input adds extra symbols in form of the end-markers.\ ropebwt, ropebwt2, IBB, and BCR are tested on the partitioned datasets. Because partDNA and the following construction algorithm are sequential, we took the sum of their runtimes, and we used the maximum of their max-rss values. In most cases, partDNA had a lower max-rss than the following BWT construction algorithm.\ eGap was tested on both, but we omit the results for eGap on the partitioned datasets because they were slower than on the single-string dataset, which is internally constructed by gSACA-K only.

\begin{table*}
    \centering  
    \caption{Used BWT construction algorithms and datasets from \href{https://www.ncbi.nlm.nih.gov/}{NCBI}.}\label{tab:bwt_construction_algortihms}
    \begin{minipage}[t]{0.51\linewidth}
    \centering
    \begin{tabular}[t]{llp{5cm}}
        approach & paper & implementation\\
        \hline
        
        \href{https://github.com/lh3/ropebwt}{ropebwt} & – & github.com/lh3/ropebwt\\

        \href{https://github.com/lh3/ropebwt2}{ropebwt2} & \cite{Li2014} & github.com/lh3/ropebwt2 \\

        \href{https://github.com/lh3/ropebwt3}{ropebwt3} & – & github.com/lh3/ropebwt3 \\

        \href{https://github.com/adlerenno/ibb}{IBB} & \cite{adler2025} & github.com/adlerenno/ibb \\

        \href{https://gitlab.com/manzai/Big-BWT}{BigBWT} & \cite{boucher2019} & gitlab.com/manzai/Big-BWT 
        \\

        \href{https://github.com/marco-oliva/r-pfbwt}{r-pfbwt} & \cite{oliva2023} & github.com/marco-oliva/r-pfbwt \\

        \href{https://github.com/y-256/libdivsufsort}{divsufsort} & \cite{fischer2017} & github.com/y-256/libdivsufsort \\

        \href{https://github.com/IlyaGrebnov/libsais}{libsais} & – & github.com/IlyaGrebnov/libsais \\

        \href{https://github.com/ddiazdom/grlBWT}{grlBWT} & \cite{navarro2023} & github.com/ddiazdom/grlBWT \\

        \href{https://github.com/felipelouza/egap}{eGap} & \cite{egidi2019} & github.com/felipelouza/egap \\

        \href{https://github.com/felipelouza/gsufsort}{gsufsort} & \cite{Louza2020} & github.com/felipelouza/gsufsort \\
        
        \href{https://github.com/giovannarosone/BCR_LCP_GSA}{BCR} &\cite{bauer2013} & github.com/giovannarosone\newline/BCR\_LCP\_GSA \\
    \end{tabular}
    \end{minipage}
    \begin{minipage}[t]{0.48\linewidth}
        \centering
    \begin{tabular}[t]{lrrr}
        dataset & $\minlength$ & average length & collection size $l$ \\ \hline

        \href{https://www.ncbi.nlm.nih.gov/datasets/genome/GCF_000001635.27/}
        {GRCm39} & – & 2,654,621,783 & 1 \\
        \href{https://www.ncbi.nlm.nih.gov/datasets/genome/GCF_000001635.27/}
        {GRCm39} & 3 & 30 & 88,252,043 \\
        \href{https://www.ncbi.nlm.nih.gov/datasets/genome/GCF_000001635.27/}
        {GRCm39} & 4 & 83 & 31,890,467 \\
        \href{https://www.ncbi.nlm.nih.gov/datasets/genome/GCF_000001635.27/}
        {GRCm39} & 5 & 215 & 12,350,256 \\
        \hline
        \href{https://www.ncbi.nlm.nih.gov/datasets/genome/GCF_000001405.40/}{GRCh38} & – & 3,049,315,783 & 1 \\
        \href{https://www.ncbi.nlm.nih.gov/datasets/genome/GCF_000001405.40/}{GRCh38} & 3 & 26 & 116,219,956 \\
        \href{https://www.ncbi.nlm.nih.gov/datasets/genome/GCF_000001405.40/}{GRCh38} & 4 & 67 & 46,529,667 \\
        \href{https://www.ncbi.nlm.nih.gov/datasets/genome/GCF_000001405.40/}{GRCh38} & 5 & 151 & 20,189,969 \\
        \hline
        \href{https://www.ncbi.nlm.nih.gov/datasets/genome/GCF_018294505.1/}{JAGHKL01} & – & 14,314,496,836 & 1 \\
        \href{https://www.ncbi.nlm.nih.gov/datasets/genome/GCF_018294505.1/}{JAGHKL01} & 3 & 40 & 356,650,569 \\
        \href{https://www.ncbi.nlm.nih.gov/datasets/genome/GCF_018294505.1/}{JAGHKL01} & 4 & 121 & 118,749,573 \\
        \href{https://www.ncbi.nlm.nih.gov/datasets/genome/GCF_018294505.1/}{JAGHKL01} & 5 & 364 & 39,284,785 \\
    \end{tabular}
    \end{minipage}
\end{table*}

\begin{figure*}
    \centering
    \begin{minipage}[t]{0.49\textwidth}
        \centering
        \includegraphics[width=\textwidth]{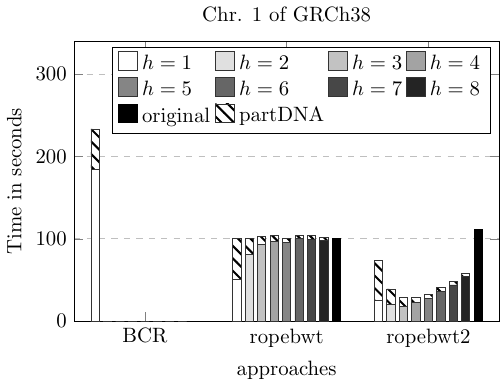}
    \end{minipage}
    \begin{minipage}[t]{0.49\textwidth}
        \centering
        \includegraphics[width=\textwidth]{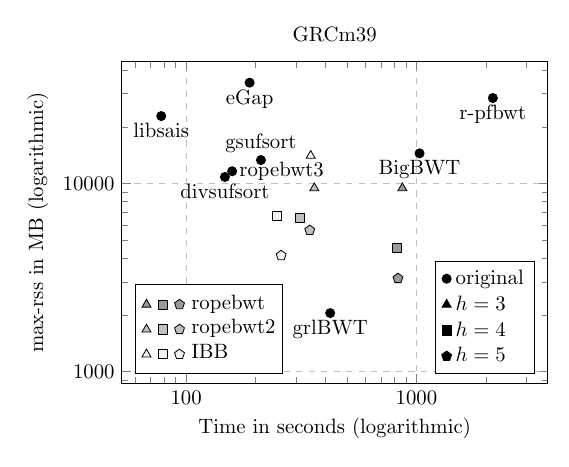}
    \end{minipage}
    \begin{minipage}[t]{0.49\textwidth}
        \centering
        \includegraphics[width=\textwidth]{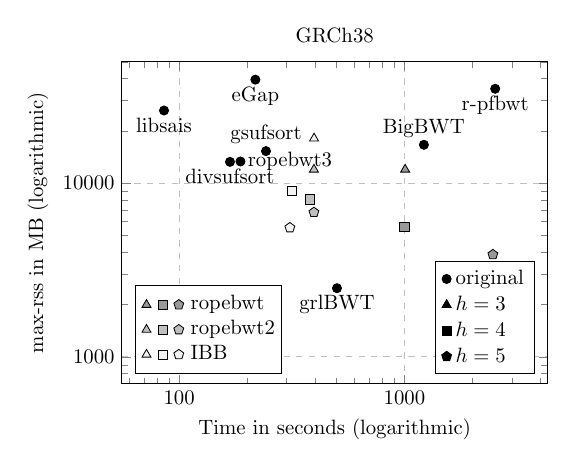}
    \end{minipage}
    \begin{minipage}[t]{0.49\textwidth}
        \centering
        \includegraphics[width=\textwidth]{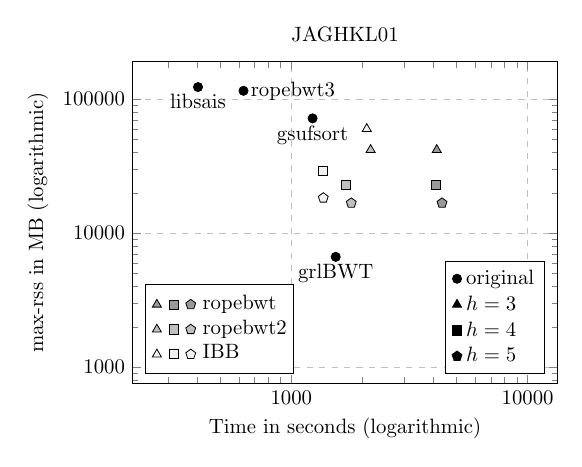}
    \end{minipage}
    \caption{BWT construction times and maximum resident set sizes (max-rss). Grey polygons in scatter plots belong to a partioned dataset: the grey tone determines the BWT construction algorithm and the number of edges the used parameter $h$, as the legends explain. Missing points mean that the construction algorithms abort or do not create an output file.}
    \label{fig:time_and_space}
\end{figure*}

In Figure~\ref{fig:time_and_space}, our test series suggest that $h \in \{3, 4, 5\}$ works best regarding construction time as shown for chromosome 1 of GRCh38. The scatter plots in Figure~\ref{fig:time_and_space} visualize the needed time and RAM usage given the file, the approach, and the length $\minlength$. An algorithm $A$ is pareto optimal if there is no other algorithm $B$ that uses both, less or equal time, which means the point of $B$ is left or equal of the point of $A$, and less or equal RAM, so the point of $B$ is equal or lower than the point of $A$. For example, eGap is not pareto optimal on the files GRCm39 and GRCh38, because divsufsort uses less time and less RAM compared to eGap, so divsufsorts point is in the left lower direction from the point of eGap.

Our test shows that libsais is the fastest construction algorithm for this class of BWT construction problems followed by divsufsort and ropebwt3.\ grlBWT uses the lowest amount of RAM.\@ Our results show, that only libsais, divsufsort, ropebwt3, grlBWT and partDNA with IBB achive pareto optimal results.\ partDNA and IBB together, especially for $h\in\{4, 5\}$, offer a new, and balanced trade-off between speed and space consumption.

\section{Conclusion}

We have presented an approach that partitions long strings of any alphabet to transform a long single-string BWT construction problem into a multi-string BWT construction problem. We have shown that our implementation partDNA designed for DNA sequences together with ropebwt2 provides a new pareto-optimum within the time-space trade-off for BWT construction. 



\newpage
\bibliographystyle{psc}
\bibliography{biblio}

\newpage
\appendix
\section{Proof of the Partition Theorem}\label{appendix:proof}

\begin{proposition}\label{lemma:first_letter}
    For any $q$ with $-1 \leq q < n$: if $q \in PSA(S)$ or $q = -1$, then \[position(q) = 0.\]
\end{proposition}

\begin{proof}
    $SA[0] = n$, because $S[n] = \$$ is by definition the smallest symbol and unique, thereby, $n$ is the index of the smallest suffix. Let $p$ be the $\min(\{t \in PSA(S) : t > q\})$ which is well defined, because $q < n$ and $n \in PSA(S)$ for any $k \geq 0$. Then $q = \Omega(p)$ and $q \in [\Omega(p), PSA(S)[p]-1]$, so $position(q) = q - \Omega(p) = q - q = 0$.
\end{proof}

\begin{proposition}\label{lemma:not_first_letter}
    For any $q$ with $0 \leq q < n$: if $q \not\in PSA(S)$, then \[word(q) = word(q-1) \text{ and } position(i - 1) = position(i) - 1.\]
\end{proposition}

\begin{proof}
    If $position(q) = 0 = q - \Omega(j)$ for a $j$, it follows $\Omega(j) = q$, so $q \in \{-1\} \cup PSA(A)$, which contradicts the assumptions. Thereby, $\Omega(j) < q < PSA(j)-1$, from this we get $\Omega(j) \leq q-1 < PSA(j)-1$, so $word(q-1) = word(q)$. $position(q-1) = q - 1 - \Omega(j) = q - \Omega(j) - 1 = position(q) - 1$.
\end{proof}

We define $W'_i = W_i + \hashtag_i$ and use the $W'_i$ in the proofs, because $BWT(W)$ contains $k$ $\hashtag$ symbols that the $W_i$ do not contain. The lexicographical order is $\hashtag_0 < \hashtag_1 < \cdots < \hashtag_k < \dollar < a$ for all $a\in\Sigma$. The advantage of using the $W'_i$ is, that their total length is equal to the length of $BWT(W)$. Thus, we can write down the suffix array for the $W'_i$ that fits to $BWT(W)$, which is not possible for the $W_i$. The advantages of using the $W_i$ in the previous part of the paper are the easier presentation of the partition of $S$ and the definition of a multi-string problem as in Figure~\ref{figure:concept}. We add the index $i$ to the $\hashtag$ symbols in order to break the tie of equal suffixes at the $\hashtag$ symbols. 

\begin{proposition}\label{lemma:ordering}
    If $S[i..] < S[j..]$, then \[W'_{word(i)}[position(i)..] < W'_{word(j)}[position(j)..].\]
\end{proposition}

\begin{proof}
    Let $c \geq 0$ be the smallest value for which either $S[i+c] \neq S[j+c]$, or $i+c \in PSA(S)$, so $S[i, i+c-1] = S[j, j+c-1]$. The idea here is, that $c$ is the distance to the positions where the tie between the suffixes starting at $i$ and $j$ breaks. 
    First, from $S[i, i+c-1] = S[j, j+c-1]$ and $S[i..] < S[j..]$, we conclude that $S[i + b..] < S[j + b..]$ for all $b$ with $0 \leq b < c$. As $PSA(S)$ contains the positions of the smallest suffixes of $S$, we conclude from $i, i+1, \dots, i + c - 1 \not\in PSA(S)$ and $S[i + b..] < S[j + b..]$ for all $b$ with $0 \leq b < c$ that $j, j+1, \dots, j + c - 1 \not\in PSA(S)$. 

    Second, by Proposition~\ref{lemma:not_first_letter}, we now get 
    \begin{align*}
        &word(i) = word(i+1) = \cdots = word(i + c - 1),\\
        &word(j) = word(j+1) = \cdots = word(j + c - 1),\\
        &position(i+c-1) = position(i+c-2) + 1 = \cdots = position(i)+ c-1,\\
        &position(j+c-1) = position(j+c-2) + 1 = \cdots = position(j) + c-1.
    \end{align*} 
    This shows, that the tie of $position(i)$ in $W'_{word(i)}$ in comparison to $position(j)$ in word $W'_{word(j)}$ is not decided before the distance $c$. In other words, if the order of the suffixes is $W'_{word(i+c-1)+1}[position(i+c-1)+1..] < W'_{word(j+c-1)+1}[position(j+c-1)+1..]$, we get $W'_{word(i)}[position(i)..] < W'_{word(j)}[position(j)..]$.

    Now, we distinguish three cases.

    First case, if $S[i+c] \neq S[j+c]$ and $i+c \not\in PSA(S)$: Like above, $j+c \not\in PSA(S)$, so we get \[W'_{word(i+c)}[position(i+c)] = S[i+c] < S[j+c] = W'_{word(j+c)}[position(j+c)].\] This is the case, when the comparison of two suffixes in $W$ can be decided without getting to the end of a word in $W$.

    Second case, if $i+c = PSA(S)[v]$ and $j+c \not\in PSA(S)$: Then \[W'_{word(i+c-1)}[position(i+c-1)+1] = \hashtag_{v} < S[j+c] = W'_{word(j+c)}[position(j+c)].\] Note that $position(i+c-1)+1 \geq 1$ and $position(i+c) = 0$ by Proposition~\ref{lemma:first_letter} because $word(i+c-1) \neq word(i+c)$. 

    Third case, if $i+c = PSA(S)[v]$ and $j+c = PSA(S)[w]$: From $S[i..] < S[j..]$, we get $v < w$ by the definition of the suffix array. Then 
    \begin{align*}
        W'_{word(i+c-1)}[position(i+c-1)+1]
        &=\hashtag_{v} \\
        &<\hashtag_{w} \\
        &= W'_{word(j+c-1)}[position(j+c-1)+1] \\
    \end{align*}
\end{proof}

\begin{theorem}\label{theorem:suffix_array}
    Let $S$ be a string of length $n$ over the alphabet $\Sigma$. Let $W$ be the collection of partitioned words $W_i' = S[\Omega(i), PSA(S)[i]-1] + \hashtag_i$ obtained from $S$. Let $l (= k + 1)$ be the size of $W$, $m (= n + l)$ be the total length of $W$, and let $SA(S)$ and $SA(W)$ be the suffix arrays of $S$ and $W$, respectively. Then, the suffix array and document array are (for all $i < m$):

    \[ SA(W)[i] = \begin{cases}
        |W_i| & 0 \leq i < l \\
        position(SA(S)[i-l]) & l \leq i < m
    \end{cases} \]

    \[ DA(W)[i] = \begin{cases}
        i & 0 \leq i < l \\
        word(SA(S)[i-l]) & l \leq i < m
    \end{cases} \]
\end{theorem}

\begin{proof}
    By construction of the words $W'_i$, the smallest $k + 1 = l$ characters in $W$ are $\hashtag_i$ and each $\hashtag_I$ occurs only once. Thereby, we get $SA(W)[i] = |W'_i| - 1 = |W_i|$, which is the position of the $\hashtag_i$ in $W'_i$, together with $DA(W)[i] = i$ for $0 \leq i < m$ by the order of the $\hashtag_i$ symbols.
    
    Next, there is a continuous block of length $n$ left in the $SA(W)$ and $DA(W)$ arrays to prove. 
    By definition of the suffix array, we get for the string $S$ \[S[SA(S)[0]..] < \cdots < S[SA(S)[n-1]..].\] 
    By Proposition~\ref{lemma:ordering}, we get the following order of the remaining $n$ suffixes of $W$: 
    \begin{align*}
        &W'_{word(SA(S)[0])}[position(SA(S)[0])..] \\
        < &W'_{word(SA(S)[1])}[position(SA(S)[1])..] \\
        < &\cdots \\
        < &W'_{word(SA(S)[n-1])}[position(SA(S)[n-1])..].
    \end{align*}

    The inequations show the order of the remaining $n$ suffixes. For example, we get that $W'_{word(SA(S)[0])}[position(SA(S)[0])..]$ is the $l$-th lowest suffix of $W$, so $SA(W)[l] = position(SA(S)[l-l])$ and $DA(W)[l] = word(SA(S)[l-l])$. The additional $-l$ within the terms $SA(S)[i-l]$ come from the fact that this order starts at position $l$ in $SA(W)$ instead of at position $0$.
\end{proof}

\setcounter{theorem}{0}
\begin{theorem}
    Let $S$ be a string of length $n$ over the alphabet $\Sigma$. Let $W$ be the collection of partitioned words $W_i' = S[\Omega(i), PSA(S)[i]-1] + \hashtag_i$ obtained from $S$. Let $l (= k + 1)$ be the size of $W$, $m (= n + l)$ be the total length of $BWT(W)$, and let $BWT(S)$ and $BWT(W)$ be the BWTs of $S$ and $W$, respectively. Then, for all $i < m$:

    \[ BWT(W)[i] = \begin{cases}
        BWT(S)[i] & 0 \leq i < l \\
        \hashtag & l \leq i < 2l \\
        BWT(S)[i - l] & 2l \leq i < m
    \end{cases} \] 
\end{theorem}

\begin{proof}

    We calculate $BWT(W)$ from $SA(W)$. For any $i < m$:

    \[ BWT(W)[i] = W'_{DA(W)[i]}[SA(W)[i] - 1].\]

    In the case that $i < l$, we get 

    \[ W'_{DA(W)[i]}[SA(W)[i] - 1] = W'_{i}[|W_i| - 1]\]

    and with the definition of $W'_i$, we get 

    \[ W'_{i}[|W_i| - 1] = S[PSA(i) - 1] = S[SA(S)[i] -1] = BWT(S)[i].\]

    In the case $i \geq l$, we get

    \[W'_{DA(W)[i]}[SA(W)[i] - 1] = W'_{word(SA(S)[i-l])}[position(SA(S)[i-l])-1].\]

    Next, if $i < 2l$, we have $i - l < l$, so $SA(S)[i-l] \in PSA(S)$. We can use Proposition~\ref{lemma:first_letter} now: $position(SA(S)[i-l]) = 0$, hence 

    \[ W'_{DA(W)[i]}[SA(W)[i] - 1] = W'_{word(SA(S)[i-l])}[0-1] = \hashtag_{word(SA(S)[i-l])} \]

    Last, if $i \geq 2l$, so $SA(S)[i-l] \not\in PSA(S)$. There is exactly one $j < l$, such that $SA(S)[i-l] \in [\Omega(j), PSA(j)-1]$. Then, $position(SA(S)[i-l]) = \Omega(j) - SA(S)[i-l]$ and $position(SA(S)[i-l]-1) = \Omega(j) - SA(S)[i-l] -1$ due to Proposition~\ref{lemma:not_first_letter}.

    \begin{align*}
        BWT(W)[i] &= W'_{word(SA(S)[i-l])}[position(SA(S)[i-l])-1] \\
        &= S[\Omega(j), PSA(j)-1][SA(S)[i-l]-1 - \Omega(j)]\\
        &= S[\Omega(j) + SA(S)[i-l]-1 - \Omega(j)] \\
        &= S[SA(S)[i-l]-1] \\
        &= BWT(S)[i-l].\\
    \end{align*}
\end{proof}

In the proofs of Theorem~\ref{theorem:partition} and \ref{theorem:suffix_array}, we have only shown the correctness of the partition using a single word $S$, but the proofs did not use the limitation that only one string $S$ was given. The presented partitioning can also transform a collection of strings $S = \{S_0, \dots, S_n\}$ into a larger collection of shorter words. The necessary changes for partitioning a collection of strings is to use a suffix array and a document array of $S$ instead of using only the suffix array of $S$ and they include a set of $\dollar$ symbols to terminate the strings and a set of $\hashtag$ symbols for partitioning the strings into words. Hereby, each $\dollar$ symbol is lexicographically larger than each $\hashtag$ symbol. There is no change necessary in the proof steps. Note that our partDNA implementation can partition a collection $S$ of strings as input.

If data is highly repetitive, similar words are next to each other in the collection $W$ after partitioning $S$. The argument is the same as for the BWT grouping similar characters together: If we divide a pattern occurring frequently in $S$, each word before such a dividing position has a common prefix of its suffix because the dividing position occurs often as well. In the example of $S$, we have twice $AAAACAAACCG$. The words $AAAAC$ are grouped together by their suffixes starting with $AAACCG$.

\section{On the Size of the Reduced Problem compared to SA-IS}

\begin{theorem}
    Each position $p > 0$ with either $S[p] = A = S[p + 1] = \cdots = S[p + d - 1]$ and $S[p + d] \not\in \{\$, A\}$ and $S[p - 1] \neq A$, or $S[p] = \$$ is a left-most S-type (LMS) position (according to the definition in SA-IS~\cite{nong2011}).
\end{theorem}

\begin{proof}
    If $S[p] = \$$: If $p = 0$ then $S = \$$. If $p > 0$, then $S[p-1] \neq \$$ and thereby, $S[p-1] > S[p]$. Then $p-1$ is L-Type and $p$ is S-Type by definition, which means that $p$ is a LMS position.

    Next, $S[p] = A$. $\$ \neq S[p - 1] \neq A$, so $S[p - 1] \in \{C, G, T\}$. We get $S[p - 1] > S[p]$, so $p-1$ is L-Type. Because $S[p + d] \not\in \{\$, A\}$, $p+d-1$ is S-Type. Finally, $S[p] = A = S[p + 1] = \cdots = S[p + d - 1]$ implies that the type of $p$ is equal to $p + 1$ is equal to $\dots$ is equal to $p+d-1$, so $p$ is S-Type and LMS-position.
\end{proof}

We conclude that our reduced problem is smaller or equal to the recursive problem of SA-IS because our reduced problem contains only one character for each position $p$ with $d > \minlength$.

\newpage
\section{List of Symbols} \label{appendix:abbreviations}

\begin{longtable}{lp{10cm}}
    Symbol & Explanation \\\hline
\endhead
    $i, j, q$ & Indices with minimal scope. \\
    $S$ & A single string of length $n$. \\
    $\Sigma$ & The alphabet for strings. \\
    $S[i..]$ & Suffix of $S$ starting at position $i$. \\
    $W$ & Collection of strings $W_i$.\\
    $SA(S)$ & Suffix array of string $S$. \\
    $SA(W)$ and $DA(W)$ & Suffix array and document array of the collection $W$. They need to be used in combination. \\
    $BWT(S)$ or $BWT(W)$ & Burrows-Wheeler transform of $S$ or $W$. \\
    $k$ & $k+1$ is a fixed value, the number of words in $W$, and the length of $PSA(S)$. \\
    $PSA(S)$ & Prefix of the suffix array $SA(S)$ of $S$ of length $k+1$. \\
    $\Omega(j)$ & Next smaller value in $PSA(S)$ or $-1$. Used to define the $W_i$. \\
    $word$ and $position$ & Functions to map a character $S[q]$ at position $q$ to its occurrence in the partition, $W_{word(q)}[position(q)]$. \\
    $l$, $m$ & $l=k+1$ is the size of set $W$; $m = n+l$ is the length of string $BWT(W)$.\\
    $\minlength$ & Minimal length of $A$ run needed to divide $S$ at that position.\\
    $ID$ & Array of word IDs.\\
    $B$ & Array containing the lengths of $A$ runs before a word.\\
    $Z$ & Array containing the lengths of $A$ runs after a word in Step 2 and Step 3. Contains runs for equal words in Step 4 and thereafter the lexicographical names for the words used in the reduced problem $R$. \\
    $R$ & String of integers, the reduced problem in partDNA. \\
    $L$ & Inverted and incremented list of the suffix array of $R$ without the suffix array position for the added end marker.\\
    $Q$ & List for induced-sorting variant of the skipped word with ID $0$.\\
\end{longtable}

\end{document}